\documentclass{article}
\usepackage{amsmath,amssymb,amsthm,graphicx,hyperref,newtxtext,newtxmath}
\usepackage[body={6.5in,9in}]{geometry}
\newtheorem{definition}{Definition}
\newtheorem{theorem}{Theorem}
\begin{document}
\title{On The Classification-Distortion-Perception Tradeoff}
%
%
\author{Dong Liu, Haochen Zhang, Zhiwei Xiong\\
\emph{University of Science and Technology of China}\\
\texttt{dongeliu@ustc.edu.cn}}
%
%
%
\maketitle              
\begin{abstract}
Signal degradation is ubiquitous and computational restoration of degraded signal has been investigated for many years. Recently, it is reported that the capability of signal restoration is fundamentally limited by the perception-distortion tradeoff, i.e. the distortion and the perceptual difference between the restored signal and the ideal ``original'' signal cannot be made both minimal simultaneously. Distortion corresponds to signal fidelity and perceptual difference corresponds to perceptual naturalness, both of which are important metrics in practice. Besides, there is another dimension worthy of consideration, namely the semantic quality or the utility for recognition purpose, of the restored signal. In this paper, we extend the previous perception-distortion tradeoff to the case of classification-distortion-perception (CDP) tradeoff, where we introduced the classification error rate of the restored signal in addition to distortion and perceptual difference. Two versions of the CDP tradeoff are considered, one using a predefined classifier and the other dealing with the optimal classifier for the restored signal. For both versions, we can rigorously prove the existence of the CDP tradeoff, i.e. the distortion, perceptual difference, and classification error rate cannot be made all minimal simultaneously. Our findings can be useful especially for computer vision researches where some low-level vision tasks (signal restoration) serve for high-level vision tasks (visual understanding).

\textbf{Keywords:} Bayes decision, classification, distortion, error rate, perception, tradeoff.
\end{abstract}

\section{Introduction}
\subsection{Motivation}
Signal degradation refers to the corruption of the signal due to many different reasons such as interference and the blend of interested signal and uninterested signal or noise, which is observed ubiquitously in practical information systems. The cause of signal degradation may be physical factors, such as the imperfectness of data acquisition devices and the noise in data transmission medium; or may be artificial factors, such as the lossy data compression and the transmission of multiple sources over the same medium at the same time. In addition, in cases where we want to enhance signal, we may assume the signal to have been somehow ``degraded,'' for example as we want to enhance the resolution of an image, we assume the image is a degraded version of an ideal ``original'' image that has high resolution \cite{dong2014learning}.

To tackle signal degradation or to fulfill signal enhancement, computational restoration of degraded signal has been investigated for many years. There are various signal restoration tasks corresponding to different degradation reasons. Taken image as example, image denoising \cite{zhang2017beyond}, image deblur \cite{su2017deep}, single image super-resolution \cite{dong2014learning}, image contrast enhancement \cite{gharbi2017deep}, image compression artifact removal \cite{dong2015compression}, image inpainting \cite{yu2018generative}, \dots, all belong to image restoration tasks.

Different restoration tasks have various objectives. Some tasks may be keen to recover the ``original'' signal as faithfully as possible, like image denoising is to recover the noise-free image, compression artifact removal is to recover the uncompressed image. Some other tasks may concern more about the perceptual quality of the restored signal, like image super-resolution is to produce image details to make the enhanced image look like ``high-resolution,'' image inpainting is to generate a complete image that looks ``natural.'' Yet some other tasks may serve for recognition or understanding purpose: for one example, an image containing a car license plate may have blur, and image deblur can achieve a less blurred image so as recognize the license plate \cite{lu2016robust}; for another example, an image taken at night is difficult to identify, and image contrast enhancement can produce a more naturally looking image that is better understood \cite{kuang2017nighttime}. Recent years have witnessed more and more efforts about the last category \cite{shermeyer2018effects,vidalmata2019bridging}.

Given the different objectives, it is apparent that a signal restoration method designed for one specific task shall be evaluated with the specific metric that corresponds to the task's objective. Indeed, the aforementioned objectives correspond to three groups of evaluation metrics:
\begin{enumerate}
  \item \emph{Signal fidelity metrics} that evaluate how similar is the restored signal to the ``original'' signal. These include all the full-reference quality metrics, such as the well-known mean-squared-error (MSE) and its counterpart peak signal-to-noise ratio (PSNR), the structural similarity (SSIM) \cite{wang2004image}, and the difference in features extracted from original signal and restored signal \cite{johnson2016perceptual}, to name a few.
  \item \emph{Perceptual naturalness metrics} that evaluate how ``natural'' is the restored signal with respect to human perception. Perceptual naturalness was evaluated by human and approximated by no-reference quality assessment methods \cite{mittal2012no,saad2012blind}. Recently, the popularity of generative adversarial network (GAN) has motivated a formulation of perceptual naturalness \cite{blau2018perception}.
  \item \emph{Semantic quality metrics} that evaluate how ``useful'' is the restored signal in the sense that it better serves for the following semantic-related analyses. For example, how well a classifier performs on the restored signal is a measure of the semantic quality. There are only a few studies about semantic quality assessment methods \cite{liu2017recognizable}.
\end{enumerate}
It is worth noting that signal fidelity metrics have dominated in the researches of signal restoration. However, is one method optimized for signal fidelity also optimal for perceptual naturalness or semantic quality? This question has been overlooked for a long while until recently. Blau and Michaeli considered signal fidelity and perceptual naturalness and concluded that both metrics cannot be optimized simultaneously \cite{blau2018perception}. Indeed, they provided a rigorous proof of the existence of the perception-distortion tradeoff: with distortion representing signal fidelity and perceptual difference representing perceptual naturalness, one signal restoration method cannot achieve both low distortion and low perceptual difference (up to a bound). This conclusion reveals the fundamental limit of the capability of signal restoration, and inspires the adoption of perceptual naturalness metrics in related tasks \cite{blau20182018,vu2018fast}.

Following the work of the perception-distortion tradeoff, in this paper, we aim to consider the three groups of metrics jointly, i.e. we want to study the relation between signal fidelity, perceptual naturalness, and semantic quality. We consider classification error rate as the representative of semantic quality, because classification is the most fundamental semantic-related analysis. We find there is indeed a tradeoff between the three metrics, which is named the classification-distortion-perception (CDP) tradeoff. In short, the CDP tradeoff claims that the distortion, perceptual difference, and classification error rate cannot be made minimal simultaneously. Our proof indicates the essential difference between the three quality metrics. In practice, it implies the adoption of semantic quality metrics instead of signal fidelity or perceptual naturalness metrics, if a signal restoration method is meant to serve for recognition purpose.
\subsection{Problem Definition}
Consider the process: $X\rightarrow Y\rightarrow \hat{X}$, where $X$ denotes the ideal ``original'' signal, $Y$ denotes the degraded signal, and $\hat{X}$ denotes the restored signal.
We formulate $X$, $Y$, and $\hat{X}$ each as a discrete random variable.
The cases of continuous random variables can be deduced in a similar manner, and thus are omitted hereafter.
The probability mass function of $X$ is denoted by $p_X(x),x\in\mathcal{X}$.
The degradation model is denoted by $P_{Y|X}$, which is characterized by a conditional mass function $p(y|x)$.
The restoration method is then denoted by $P_{\hat{X}|Y}$ and characterized by $p(\hat{x}|y)$.

We are interested in classifying the signal into two categories in this paper. Thus, we assume each sample of the original signal belongs to one of two classes: $\omega_1$ or $\omega_2$. The \emph{a priori} probabilities and the conditional mass functions are assumed to be known as $P_1,P_2$ and $p_{X1}(x),p_{X2}(x)$, respectively. In other words, $X$ follows a two-component mixture model: $p_X(x)=P_1p_{X1}(x)+P_2p_{X2}(x)$. Accordingly, $Y$ follows the model: $p_Y(y)=P_1p_{Y1}(y)+P_2p_{Y2}(y)$, and $\hat{X}$ follows the model: $p_{\hat{X}}(\hat{x})=P_1p_{\hat{X}1}(\hat{x})+P_2p_{\hat{X}2}(\hat{x})$, where
\begin{eqnarray}
  p_{Yi}(y) &=& \sum_{x\in\mathcal{X}}{p(y|x)p_{Xi}(x)}, i=1,2 \\
  \nonumber p_{\hat{X}i}(\hat{x}) &=& \sum_{y\in\mathcal{Y}}{p(\hat{x}|y)p_{Yi}(y)}\\
  &=& \sum_{y}\sum_{x}{p(\hat{x}|y)p(y|x)p_{Xi}(x)}, i=1,2
\end{eqnarray}

A binary classifier can be denoted by
\begin{equation}\label{cx}
  c(t)=c(t|\mathcal{R})=\begin{cases}
         \omega_1, & \mbox{if } t\in\mathcal{R} \\
         \omega_2, & \mbox{otherwise}
       \end{cases}
\end{equation}
If we apply this classifier on the original signal $X$, we shall achieve an error rate
\begin{equation}\label{errorrate}
  \varepsilon(X|c)=\varepsilon(X|\mathcal{R})=P_2\sum_{x\in\mathcal{R}}{p_{X2}(x)}+P_1\sum_{x\notin\mathcal{R}}{p_{X1}(x)}
\end{equation}

The optimal classifier is defined as the classifier that achieves the minimal error rate for a given signal, e.g.
$c_{X}^*=\arg\min_{c}{\varepsilon(X|c)}$.
According to the Bayes decision rule (see \cite{fukunaga1990introduction} for proof), the optimal classifier shall be
\begin{equation}\label{bayes_classifier}
  c_X^*=c(\cdot|\mathcal{R}_X^*),\mbox{where }\mathcal{R}_X^*=\{x|P_1p_{X1}(x)\geq P_2p_{X2}(x)\}
\end{equation}
which leads to the minimal error rate, a.k.a. the Bayes error rate
\begin{equation}\label{bayes_errorrate}
\begin{aligned}
  \epsilon(X)&=\min_{c}{\varepsilon(X|c)}=\varepsilon(X|\mathcal{R}_X^*)\\
            &=\sum_{x}{\min[P_1p_{X1}(x),P_2p_{X2}(x)]} \\
            &=\frac{1}{2} - \frac{1}{2}\sum_{x}{|P_1p_{X1}(x)-P_2p_{X2}(x)|}
\end{aligned}
\end{equation}
\subsection{Main Theorems}
We prove two versions of the CDP tradeoff. For the first version, we consider using a predefined classifier $c_0=c(\cdot|\mathcal{R}_0)$ on the restored signal. This leads to
\begin{definition}
The classification-distortion-perception (CDP) function is
\begin{equation}\label{cdp_function}
  C(D,P)=\min_{P_{\hat{X}|Y}}\varepsilon(\hat{X}|c_0),\mbox{subject to }\mathbb{E}[\Delta(X,\hat{X})]\leq D,d(p_X,p_{\hat{X}})\leq P
\end{equation}
where $\mathbb{E}$ is to take expectation, $\Delta(\cdot,\cdot): \mathcal{X}\times\hat{\mathcal{X}}\rightarrow \mathbb{R}^+$ is a function to measure distortion between the original and the restored signals, and $d(\cdot,\cdot)$ is a function to measure the difference between two probability mass functions, which is claimed to be indicative for perceptual difference \cite{blau2018perception}.
\end{definition}
\begin{theorem}\label{cdp_theorem}
Consider (\ref{cdp_function}), if $d(\cdot,q)$ is convex in $q$, then $C(D,P)$ is
\begin{enumerate}
  \item monotonically non-increasing,
  \item convex in $D$ and $P$.
\end{enumerate}
\end{theorem}
Note that the convexity of the perceptual difference is assumed, which is claimed to be satisfied by a large number of commonly used difference functions, including any f-divergence (e.g. Kullback-Leibler divergence, total variation, Hellinger) and the R{\'e}nyi divergence \cite{csiszar2004information,van2014renyi}.

For the second version, we consider using the optimal classifier on the restored signal, i.e. the classifier is adaptive to the restored signal. According to the Bayes decision rule, we are actually considering the Bayes error rate of $\hat{X}$. This leads to
\begin{definition}
The \textbf{strong} classification-distortion-perception (SCDP) function is
\begin{equation}\label{scdp_function}
  C_S(D,P)=\min_{P_{\hat{X}|Y}}\epsilon(\hat{X}),\mbox{subject to }\mathbb{E}[\Delta(X,\hat{X})]\leq D,d(p_X,p_{\hat{X}})\leq P
\end{equation}
\end{definition}
\begin{theorem}\label{scdp_theorem}
Consider (\ref{scdp_function}), $C_S(D,P)$ is monotonically non-increasing.
\end{theorem}
\subsection{Paper Organization}
In the following sections, we first give some properties of the classification error rate, especially the Bayes error rate, which will be helpful in our proofs of the main theorems. Then we prove the two theorems one by one. Discussion and conclusion are finally presented.
\section{Properties of the Classification Error Rate}
\subsection{Classification Error Rate is Linear}
\begin{theorem}\label{linearity_theorem}
Let $U$ follow a two-component mixture model: $p_U(u)=P_1p_{U1}(u)+P_2p_{U2}(u)$, similarly $V$ follow: $p_V(v)=P_1p_{V1}(v)+P_2p_{V2}(v)$. Let $W$ be the random variable with $p_W(w)=\lambda p_U(w)+(1-\lambda)p_V(w)$ where $0\leq \lambda\leq 1$. Let $c_0$ be a fixed classifier, then
\begin{equation}\label{errorrate_linear}
\varepsilon(W|c_0)=\lambda\varepsilon(U|c_0)+(1-\lambda)\varepsilon(V|c_0)
\end{equation}
\end{theorem}
\begin{proof}
As $c_0$ is a fixed classifier, it can be denoted in general by $c_0=c(\cdot|\mathcal{R}_0)$. Then we have
\begin{eqnarray}
  \varepsilon(U|c_0)&=&P_2\sum_{u\in\mathcal{R}_0}p_{U2}(u)+P_1\sum_{u\notin\mathcal{R}_0}p_{U1}(u) \\
  \varepsilon(V|c_0)&=&P_2\sum_{v\in\mathcal{R}_0}p_{V2}(v)+P_1\sum_{v\notin\mathcal{R}_0}p_{V1}(v)
\end{eqnarray}
Thus
\begin{equation}
\begin{aligned}
\varepsilon(W|c_0)&=P_2\sum_{w\in\mathcal{R}_0}p_{W2}(w)+P_1\sum_{w\notin\mathcal{R}_0}p_{W1}(w) \\
                  &=P_2\sum_{w\in\mathcal{R}_0}\left[\lambda p_{U2}(w)+(1-\lambda)p_{V2}(w)\right]+P_1\sum_{w\notin\mathcal{R}_0}\left[\lambda p_{U1}(w)+(1-\lambda)p_{V1}(w)\right] \\
                  &=\lambda\left[P_2\sum_{w\in\mathcal{R}_0}p_{U2}(w)+P_1\sum_{w\notin\mathcal{R}_0}p_{U1}(w)\right]
                  +(1-\lambda)\left[P_2\sum_{w\in\mathcal{R}_0}p_{V2}(w)+P_1\sum_{w\notin\mathcal{R}_0}p_{V1}(w)\right]\\
                  &=\lambda\varepsilon(U|c_0)+(1-\lambda)\varepsilon(V|c_0)
\end{aligned}
\end{equation}
\end{proof}
\subsection{Bayes Error Rate is Concave}
\begin{theorem}\label{ber_concave_theorem}
Let $U$, $V$, and $W$ be defined as in Theorem \ref{linearity_theorem}, then
\begin{equation}\label{concave}
\epsilon(W)\geq\lambda\epsilon(U)+(1-\lambda)\epsilon(V)
\end{equation}
\end{theorem}
\begin{proof}
$c^*_W$ denotes the optimal classifier for $W$, then $\epsilon(W)=\varepsilon(W|c^*_W)$. According to (\ref{errorrate_linear}) we have $\epsilon(W)=\lambda\varepsilon(U|c^*_W)+(1-\lambda)\varepsilon(V|c^*_W)$. Note that $\epsilon(U)=\min\varepsilon(U|\cdot)$ and $\epsilon(V)=\min\varepsilon(V|\cdot)$. Thus $\epsilon(W)\geq\lambda\epsilon(U)+(1-\lambda)\epsilon(V)$.
\end{proof}
\subsection{Bayes Error Rate is Non-Decreasing}
\begin{theorem}\label{ber_non_decrease}
Let the process of $X\rightarrow Y$ be denoted by $P_{Y|X}$, which is characterized by a conditional mass function $p(y|x)$, then $\epsilon_Y\geq \epsilon_X$ and $\epsilon_Y= \epsilon_X$ if and only if $p(y|x)$ satisfies: $\forall x_1\in\mathcal{R}^+,\forall x_2\in\mathcal{R}^-,\forall y,p(y|x_1)p(y|x_2)=0$, where $\mathcal{R}^+=\{x|P_1p_{X1}(x)> P_2p_{X2}(x)\}$, and $\mathcal{R}^-=\{x|P_1p_{X1}(x)< P_2p_{X2}(x)\}$. Note that $\mathcal{R}^+$ is slightly different from $\mathcal{R}_X^*$ defined in (\ref{bayes_classifier}).
\end{theorem}
\begin{proof}
\begin{equation}
\begin{aligned}
  \epsilon_Y&=\sum_{y}{\min[P_1p_{Y1}(y),P_2p_{Y2}(y)]} \\
            &=\frac{1}{2} - \frac{1}{2}\sum_{y}{|P_1p_{Y1}(y)-P_2p_{Y2}(y)|} \\
            &=\frac{1}{2} - \frac{1}{2}\sum_{y}{\left|P_1\sum_{x}{p(y|x)p_{X1}(x)}-P_2\sum_{x}{p(y|x)p_{X2}(x)}\right|}\\
            &= \frac{1}{2} - \frac{1}{2}\sum_{y}{\left|\sum_{x}p(y|x)[P_1{p_{X1}(x)}-P_2{p_{X2}(x)}]\right|}\\
            &\geq \frac{1}{2} - \frac{1}{2}\sum_{y}{\sum_{x}p(y|x)|P_1{p_{X1}(x)}-P_2{p_{X2}(x)}|}\\
            &=\frac{1}{2} - \frac{1}{2}\sum_{x}{|P_1{p_{X1}(x)}-P_2{p_{X2}(x)}|\sum_{y}p(y|x)}\\
            &=\frac{1}{2} - \frac{1}{2}\sum_{x}{|P_1{p_{X1}(x)}-P_2{p_{X2}(x)}|}=\epsilon_X
\end{aligned}
\end{equation}
When $\epsilon_Y=\epsilon_X$, for any $y$, we need to have
\begin{equation}
\left|\sum_{x}p(y|x)[P_1{p_{X1}(x)}-P_2{p_{X2}(x)}]\right|=\sum_{x}p(y|x)|P_1{p_{X1}(x)}-P_2{p_{X2}(x)}|
\end{equation}
which is equivalent to: all the $x$'s that satisfy $p(y|x)\neq 0$ shall have either $P_1{p_{X1}(x)}-P_2{p_{X2}(x)}\geq 0$ or $P_1{p_{X1}(x)}-P_2{p_{X2}(x)}\leq 0$. The condition is further equivalent to: the $x$'s that satisfy $p(y|x)\neq 0$ shall be either all in $\mathcal{R}^+\cup\mathcal{R}^0$, or all in $\mathcal{R}^-\cup\mathcal{R}^0$, where $\mathcal{R}^0=\{x|P_1p_{X1}(x)=P_2p_{X2}(x)\}$. In other words, $\forall x_1\in\mathcal{R}^+,\forall x_2\in\mathcal{R}^-,p(y|x_1)p(y|x_2)=0$.
\end{proof}

We can compare Theorem \ref{ber_non_decrease} with the data processing theorem in the information theory: consider the process of $X\rightarrow Y$ as a deterministic function $Y=f(X)$, then $I(X;Y)\leq H(X)$, and $I(X;Y)=H(X)$ if and only if $f$ is invertible \cite{cover2012elements}. That says, the information quantity we have about the source $X$ is non-increasing after data processing. Similarly, Theorem \ref{ber_non_decrease} claims that the Bayes error rate is non-decreasing after data processing, \emph{because we lose information, at best not}. Moreover, not only invertible function satisfies the condition required in Theorem \ref{ber_non_decrease}, but also a large group of non-invertible functions as well as probabilistic mappings satisfy the condition, which is quite different from the data processing theorem. In other words, \emph{we may lose information but that information loss may not affect classification}.

\section{Proof of the CDP Tradeoff (Theorem \ref{cdp_theorem})}
\begin{proof}
  For the first point, simply note that when increasing $D$ or $P$, the feasible domain of $P_{\hat{X}|Y}$ is enlarged; as $C(D,P)$ is the minimal value of $\varepsilon(\hat{X}|c_0)$ over the feasible domain, and the feasible domain is enlarged, the minimal value will not increase.

  For the second point, it is equivalent to prove:
  \begin{equation}\label{cdp_convex}
    \lambda C(D_1,P_1)+(1-\lambda)C(D_2,P_2)\geq C(\lambda D_1+(1-\lambda)D_2,\lambda P_1+(1-\lambda)P_2)
  \end{equation}
  for any $\lambda\in[0,1]$. First, let $\mu(\hat{x}|y)$ (resp. $\nu(\hat{x}|y)$) denote the optimal restoration method under constraint $(D_1,P_1)$ (resp. $(D_2,P_2)$), and $\hat{X}_{\mu}$ (resp. $\hat{X}_{\nu}$) be the restored signal, i.e.
  \begin{equation}\label{muxy}
    \varepsilon(\hat{X}_{\mu}|c_0)=\min_{P_{\hat{X}|Y}}\varepsilon(\hat{X}|c_0),\mbox{subject to }\mathbb{E}[\Delta(X,\hat{X})]\leq D_1,d(p_X,p_{\hat{X}})\leq P_1
  \end{equation}
  \begin{equation}\label{nuxy}
    \varepsilon(\hat{X}_{\nu}|c_0)=\min_{P_{\hat{X}|Y}}\varepsilon(\hat{X}|c_0),\mbox{subject to }\mathbb{E}[\Delta(X,\hat{X})]\leq D_2,d(p_X,p_{\hat{X}})\leq P_2
  \end{equation}
  Then the left hand side of (\ref{cdp_convex}) becomes
  \begin{equation}\label{cdp_convex_left}
    \lambda\varepsilon(\hat{X}_{\mu}|c_0)+(1-\lambda)\varepsilon(\hat{X}_{\nu}|c_0)=\varepsilon(\hat{X}_{\lambda}|c_0)
  \end{equation}
  where we have used Theorem \ref{linearity_theorem} and $\hat{X}_{\lambda}$ denotes the restored signal corresponding to $p_{\lambda}(\hat{x}|y)=\lambda\mu(\hat{x}|y)+(1-\lambda)\nu(\hat{x}|y)$. Let $D_{\lambda}=\mathbb{E}[\Delta(X,\hat{X}_{\lambda})]$, $P_{\lambda}=d(p_X,p_{\hat{X}_{\lambda}})$, then by definition
  \begin{equation}\label{cdp_convex_by_def}
    \varepsilon(\hat{X}_{\lambda}|c_0)\geq C(D_{\lambda},P_{\lambda})
  \end{equation}
  Next, as $d(\cdot,\cdot)$ in (\ref{cdp_function}) is convex in its second argument, we have
  \begin{equation}
  \begin{aligned}
    P_{\lambda}&=d(p_X,\lambda p_{\hat{X}_{\mu}}+(1-\lambda) p_{\hat{X}_{\nu}})\\
               &\leq \lambda d(p_X,p_{\hat{X}_{\mu}})+(1-\lambda) d(p_X,p_{\hat{X}_{\nu}})\\
               &\leq \lambda P_1+(1-\lambda) P_2
  \end{aligned}
  \end{equation}
  the last inequality is due to (\ref{muxy}) and (\ref{nuxy}). Similarly, we have
  \begin{equation}
  \begin{aligned}
  D_{\lambda}&=\mathbb{E}[\Delta(X,\hat{X}_{\lambda})]\\
             &=\mathbb{E}_{Y}\mathbb{E}[\Delta(X,\hat{X}_{\lambda})|Y]\\
             &=\mathbb{E}_{Y}[\lambda\mathbb{E}[\Delta(X,\hat{X}_{\mu})|Y]+(1-\lambda)\mathbb{E}[\Delta(X,\hat{X}_{\nu})|Y]]\\
             &=\lambda\mathbb{E}[\Delta(X,\hat{X}_{\mu})]+(1-\lambda)\mathbb{E}[\Delta(X,\hat{X}_{\nu})]\\
             &\leq \lambda D_1+(1-\lambda) D_2
  \end{aligned}
  \end{equation}
  the last inequality is again due to (\ref{muxy}) and (\ref{nuxy}). Finally, note that $C(D,P)$ is non-increasing with respect to $D$ and $P$,
  \begin{equation}\label{cdp_convex_by_dp}
    C(D_{\lambda},P_{\lambda})\geq C(\lambda D_1+(1-\lambda) D_2,\lambda P_1+(1-\lambda) P_2)
  \end{equation}
  Combining (\ref{cdp_convex_left}), (\ref{cdp_convex_by_def}), and (\ref{cdp_convex_by_dp}), we have (\ref{cdp_convex}).
\end{proof}

\section{Proof of the CDP Tradeoff (Theorem \ref{scdp_theorem})}
\begin{proof}
Simply note that when increasing $D$ or $P$, the feasible domain of $P_{\hat{X}|Y}$ is enlarged; as $C_S(D,P)$ is the minimal value of $\epsilon(\hat{X})$ over the feasible domain, and the feasible domain is enlarged, the minimal value will not increase.
\end{proof}
\section{Discussion and Conclusion}
We would like to mention the difference between Theorem \ref{cdp_theorem} and Theorem \ref{scdp_theorem}. Theorem \ref{scdp_theorem} is more fundamental as we deal with the theoretically minimal error rate of the restored signal. However in practice, this error rate is not achievable if the degradation model is \emph{unknown}. Clearly, if $p(y|x)$ is not available, we cannot make any meaningful conclusion regarding the mass function $p_Y(y)$, which prohibits the search for the optimal restoration method together with the optimal classifier. From a practical perspective, we usually adopt a fixed classifier (for example the classifier trained by some samples of the original signal) and adjust the restoration method only. On the other hand, if the degradation model is \emph{known}, then it is possible to consider the optimal classifier for the degraded signal $Y$ directly: actually it is better in theory to consider $Y$ instead of $\hat{X}$ because we have confirmed that $\epsilon(\hat{X})\geq \epsilon(Y)$ (Theorem \ref{ber_non_decrease}). In other words, signal restoration has no use to improve the classification accuracy as long as the degradation model is known. According to these analyses, Theorem \ref{scdp_theorem} is less appealing in practice.

Note that we do not prove the convexity of the strong CDP function, as we have done for the CDP function in Theorem \ref{cdp_theorem}. This is due to the essential difference between classification error rate with a fixed classifier and Bayes error rate: the former is linear and the latter is concave (Theorems \ref{linearity_theorem} and \ref{ber_concave_theorem}). Note that the distortion is also linear but the perceptual difference is convex. We suspect the strong CDP function may be not convex, which is to be confirmed in the future.

Our findings can be useful especially for computer vision researches where some
low-level vision tasks (signal restoration) serve for high-level vision tasks (visual understanding). If the degradation model is known, we recommend directly classifying the degraded signal without any restoration; at this time the classifier can be trained by samples that are simulated with the known degradation on the original signal. If the degradation model is unknown, we recommend using a fixed classifier, which can be trained for example using samples of the original signal; meanwhile, we recommend searching for the restoration method with the classification error rate as the objective (or one of the objectives) of optimization. This strategy is clearly different from previous works that optimize for various kinds of distortion metrics for improving the classification performance.

 
\end{document}